\documentclass[10pt,showpacs,floatfix,letterpaper,aps,pra,reprint, showkeys]{revtex4-1}
\usepackage{florjanc}
\usepackage{appendix}

\newtheorem{innerlemmaproof}{Proof of Lemma}
\newenvironment{lemmaproof}[1]
  {\innerlemmaproof}
  {\endinnerlemmaproof}

\begin{document}
\title{In-situ Adaptive Encoding for \\ Asymmetric Quantum Error Correcting Codes}
\author{Jan Florjanczyk} \email{florjanc@usc.edu}
\author{Todd A. Brun} \email{tbrun@usc.edu}
\affiliation{Center for Quantum Information Science and Technology, \\
Communication Sciences Institute, Department of Electrical Engineering, \\
University of Southern California, Los Angeles, CA 90089, USA.}
\date{\today}


\begin{abstract}
We present techniques that improve the performance of asymmetric stabilizer codes in the presence of unital channels with unknown parameters. Our method estimates the channel parameters using information recovered from syndrome measurements during standard stabilizer quantum error correction and adaptively realigns the codespace to minimize the uncorrectable error rate. We find that for dephasing channels parametrized by a single angle our scheme yields lifetimes dominated by the bit-flip error rate for which the asymmetric code has an improved distance. In the case of general unital channels we are able to learn and exploit orientations of the channel that yield a constant improvement to the code lifetime. In both cases, since our method is adaptive and online, we are able to model the effect of drift in the channel parameters.
\end{abstract}
\maketitle

\section{Introduction}

The last two decades of research in quantum computing have yielded remarkable advances in quantum error correction and fault-tolerant quantum computing. Error correction is a necessary component for any quantum device exposed to an environment which, by necessity of interacting with the quantum system, will introduce noise and decoherence. Error correction is considered successful when the action of the environment is minimized in some way, either by encoding information in a subspace (as is the case with decoherence-free subspace methods~\cite{dfs}), or by actively countering it via recovery operations as is the case with stabilizer quantum error correction~\cite{gottesman}.

An $[[n, k, d]]$ quantum stabilizer code protects $k$ logical qubits by encoding them into $n$ physical qubits. Let $\mcal{P}_n$ be the group of all tensor products of Pauli operators $X$, $Y$, and $Z$ and the identity $I$ on $n$ qubits (possibly with overall factors of $\pm 1$ or $\pm i$).  The code itself is defined by an Abelian subgroup $S \subset \mcal{P}_n$ such that $S$ is generated by a set of $n-k$ commuting \emph{stabilizer generators} $\lb g_s \rb$ and $S$  does not contain $-I$. We say a quantum state $\ket{\psi}$ is in the \emph{codespace} of the code if $g_s \ket{\psi} = \ket{\psi}$ for all $g_s$. A quantum state can accumulate an error by interacting with the environment. By measuring all of the stabilizer generators $g_s$, we can project the system to a state $E \ket{\psi}$ , where $E$ is an element of $\mcal{P}_n$.  The results of these ${n-k}$ measurements collectively form the \emph{syndrome} of $E$. Many different errors $E$ can result in the same syndrome and we denote by $d$ the distance of the code if it can unambiguously correct errors of weight $\lfloor (d-1)/2 \rfloor$ where the \emph{weight} of the error is the number of non-identity elements in $E \in \mcal{P}_n$.  For a general noise process, the state may be projected into a mixture of states with the same error syndrome; in a well-designed code, this mixture will be dominated by the most likely error, which is used for the correction.

Generally, a stabilizer code is characterized only in terms of the weights of its correctable errors, and not in terms of the specific types of errors produced by a quantum channel. It is natural therefore to consider error correction schemes that are designed to be effective for a specific type of error, rather than for all or errors below a given weight. Many methods for doing this have already been considered including designing new error correction procedures through direct optimization~\cite{lidar}, concatenating repetition codes with Calderbank-Shor-Steane (CSS) codes in the presence of dominant dephasing noise~\cite{shor}, and using asymmetric quantum codes to combat asymmetric noise~\cite{asymmetric1, asymmetric2, asymmetric3, asymmetric4, asymmetric5, cafaro1, cafaro2}.

Our focus in this paper will be this last case and we will design a system that exploits asymmetric quantum codes in the presence of sources of noise that exhibit a bias towards one type of error ($X$) over the other ($Z$).  (The reverse bias can be easily dealt with by switching the roles of the $X$ and $Z$ bases.)

\begin{defin}[Asymmetric quantum stabilizer code~\cite{asymmetric2}]
The asymmetric quantum stabilizer code denoted by $[[n, k, d_X/d_Z]]$ encodes $k$ logical qubits into $n$ physical qubits and can correct errors with up to $t_x = \lfloor (d_x - 1)/2 \rfloor$ $X$ operators and $t_z = \lfloor (d_z - 1)/2 \rfloor$ $Z$ operators.
\end{defin}

We will use as our test case the $[[15, 1, 7/3]]$ shortened Reed-Muller code~\cite{guillaume} which is the smallest asymmetric CSS code, and has the added benefit of allowing transversal non-Clifford operations~\cite{largeblockcodes}. Another CSS construction~\cite{bch} uses Bose-Chaudhuri-Hocquenghem (BCH) codes to form a $[[31, 6, 7/5]]$ asymmetric code which we will study in the second part of this work as well. It should be noted that our method is compatible with all asymmetric stabilizer codes, and has the advantage that our feedback control scheme is only a function of the noise channel and the code distances. As such, the analysis contained in this work applies to all codes that express asymmetric $X$ and $Z$ distances (though we only consider CSS codes in our examples).

\section{Quantum channels}
\label{sec:quantumchannels}

In general, a quantum operation $\Lambda$ is any map on the set of bounded positive operators on a Hilbert space that is completely positive i.e.: $\Lambda \ox I_n$ is positive for identity maps of any dimension $n$. Every such operation can be written in terms of Kraus operators $A_i$
\begin{equation}
	\Lambda \lp \rho \rp = \sum_i A_i^{\dag} \rho A_i.
\end{equation}
Additionally, if the operation preserves the trace, that is to say $\Tr \ls \Lambda \lp \rho \rp \rs = \Tr \ls \rho \rs$ for any state $\rho$, then it is known as a \emph{quantum channel}.

\begin{defin}[Unital quantum channel]
A unital quantum channel $\Lambda$ is a completely positive trace-preserving map on quantum states for which the maximally mixed state is a fixed point, that is 
\begin{equation}
	\Lambda(I / d) = I / d
\end{equation}
when $\Lambda$ acts on states in a Hilbert space of dimension $d$.
\end{defin}

A unital channel can equivalently be expressed as the convex combination of unitary channels~\cite{ruskai} and for this reason they often appear as the result of a partial trace after a joint unitary evolution over a quantum state and an environment. A common and simple restriction to unital channels is the Pauli channel.

\begin{defin}[Pauli channel]
A Pauli channel is a channel over $n$ qubits that can be written in the form
\begin{equation}
  \mcal{E} \lp \rho \rp = (1-p) \rho + \sum_i p_i A_i \rho A_i
\end{equation}
for $\sum_i p_i = p$ and $A_i \in \lb I, X, Y, Z \rb^{\ox n}$.
\end{defin}

In this work we will consider $n$ identical single-qubit channels acting on $n$ distinct qubits. The state of each qubit can be represented as a vector $\vec{r} \in \mbR^3$ contained inside the Bloch sphere~\cite{nielsenchuang}. We write the qubit state $\rho$ as
\begin{equation}
	\rho = \frac{I + \vec{r} \cdot \vec{\s}}{2}
\end{equation}
where $\vec{\s} = \ls X, Y, Z \rs$. Unital channels that act on single qubits can be described by their action on the vector $\vec{r}$
\begin{equation}
	\Lambda \lp \rho \rp = \frac{I + (M \vec{r}) \cdot \vec{\s}}{2}
\end{equation}
where the \emph{Bloch matrix} $M$ is any real $3 \times 3$ matrix that respects the Fujiwara-Algoet conditions~\cite{fujiwara1, fujiwara2, universalset}. We now introduce a subset of unital channels that will be the focus of our analysis.

\begin{defin}[Oriented Pauli channel]
\label{def:orientedpauli}
An \emph{oriented Pauli channel} $M$ is the result of a unitary channel $Q_U$, followed by a Pauli channel $D$, and the inverse unitary channel $Q_U^T$.
\end{defin}

We are motivated to use this definition due to the fact that any unital channel Bloch matrix $M$ can always be decomposed using the polar decomposition $M = Q_VP$ where $Q_V \in SO(3)$ and $P$ is positive semi-definite. When the action of the channel can be described entirely by $P$ (i.e.: $Q_V=I$), we have the result of Lemma~\ref{lem:Mrewrite} found the Appendix,
\begin{equation}
	M = (1-2p) I + 2p Q_U^T \ls \begin{array}{ccc}
      k_1 & 0 & 0 \\
      0 & k_2 & 0 \\
      0 & 0 & k_3
    \end{array} \rs Q_U .
    \label{eqn:oriented_Pauli_m}
\end{equation}
where $Q_U \in SO(3)$. In this picture, $p$ is the probability of acting non-identically on the qubit and $k_i = p_i / \sum_j p_j$ where $p_i$ is are the $X$, $Y$, and $Z$ error rates in some non-standard basis $Q_U$. Note also that when $Q_U=I$ and $p<1/2$, this reduces to a single-qubit Pauli channel.

Definition~\ref{def:orientedpauli} also lends itself a description in terms of CPTP maps. In other words, an oriented Pauli channel represented by the map $\Lambda_M$ is the result of a unitary channel $\Lambda_U$, followed by a Pauli channel $\Lambda_D$, and the inverse unitary channel $\Lambda_{U^{\dag}}$. Additionally, we note that the polar decomposition $M = Q_VP$ shows that any general unital channel can be written as the composition of an oriented Pauli channel with a unitary channel $\Lambda_V$.

Characterizing quantum channels is routinely done by quantum tomography which has yielded powerful methods for faithfully reconstructing process matrices with a tractable number of samples~\cite{kosut, flammia}. Tomography methods such as compressed sensing require the preparation of resource-intensive randomized quantum states and measurements. A full reconstruction of all the Kraus operators constituting a quantum operation might not be needed for every task however. The goal of more recent work done in modeling error channels has been to allow for efficient simulation of fault-tolerance~\cite{pta1, pta2, pta3, pta4} by approximating physical error models under the Pauli twirl operation. In~\cite{emerson2007symmetrized}, the authors perform a coarse-grained averaging of the qubits by applying random Clifford operators and qubit permutations which allow for efficient extraction of channel parameters such as the probability of phase-flip errors of any given weight.

We specifically consider characterizing and exploiting the asymmetry of a quantum channel in-situ. We will estimate channel parameters by using the existing error correction apparatus, which has the benefit of not causing interruptions to ongoing computations. It also allows us to account for noise channels that may be changing over time. The authors of~\cite{ferrie} already do this with great success in the presence of dephasing noise followed by an unknown $X$-axis rotation. The recent experiments of~\cite{iontrap} demonstrate a similar method in a system of nine physical qubits implementing the $5$-bit repetition code. Again, the authors are able to compensate for time-varying parameters in the noise channel but without the restriction that the channel, or the corrective controls, be identical on all qubits. Like these works, our goal in learning the channel parameters will be to improve the performance of error correction.

\section{Control scheme}
\label{sec:noisecontrolscheme}

We now present a model for reducing the rate of uncorrectable errors in two different parameterized noise channels. Our model of feedback control will be to apply a coherent rotation $\hat{U}_t$ to the physical qubits and to modify our stabilizer measurements by the same unitary, as shown in Figure~\ref{fig:circuit}. Specifically, for stabilizer generators $g_s$,
\begin{equation}
  g_s \rightarrow \lp \hat{U}_t \rp^{\ox n} g_s \lp \hat{U}_t^{\dag} \rp^{\ox n} .
  \label{eqn:g_smod}
\end{equation}
In the single-parameter case we describe below, we will use this construction to effectively rotate a dephasing channel into a bit-flip channel. Later, when we examine a multi-parameter model, this same construction will counteract the unitary rotation $\Lambda_U$ of an oriented Pauli channel.

\begin{figure}
  \centering
    \includegraphics[width=0.4\textwidth]{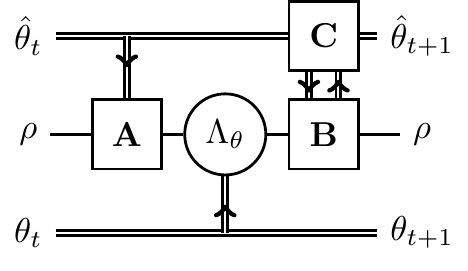}
  \caption{Our model of feedback control for optimal use of the asymmetric code uses a real-time estimate of the parametrized noise channel $\Lambda_{\theta}$ to modify the codespace and stabilizers. (\textbf{A}) We rotate all qubits in the codespace by $\hat{U}_t$. (\textbf{B}) Stabilizer error correction using the modified stabilizers of Eq.~\eqref{eqn:g_smod}. (\textbf{C}) The estimator updates the form of the stabilizer measurements. Following the syndrome extraction step, the estimator also updates the parameter estimate, symbolized by the two directions of the flow of information. Quantum information and classical information are distinguished in this circuit with single and double lines, respectively. We note, however, that the classical information of $\hat{\theta}_t$ is actually a complete encoding of the estimator of $\theta_t$. This can mean, as it does in our case, that the estimator is encoded as a point-wise approximation of the Bayesian prior distribution of $\theta_t$. Additionally, A, B, C, are error-free.}
  \label{fig:circuit}
\end{figure}

In this paper we will approximate the rate of uncorrectable errors by counting all errors whose weight exceeds the minimum number of correctable errors of the code.  (In practice, it may be possible to improve on our results by considering correctable errors of all weights, since most codes can correct many higher-weight errors as well.) Thus, the uncorrectable error rate for an $[[n, k, d_X/d_Z]]$ asymmetric code can be found by summing the probability of all possible errors that violate the distances $d_X$ and $d_Z$ yet act on the fewest number of qubits. We will denote by $\mcal{T}$ the set of error weights which violate the constraints in this way. The individual qubit errors can occur on any of the $n$ qubits, which we account for by the multinomial coefficient 
\begin{equation}
	{n \choose w_x, w_y, w_z} = \frac{n!}{(n-w_x-w_y-w_z)! w_x! w_y! w_z!} 
\end{equation}
where $w_x+w_y+w_z = w \leq n$. Altogether,
\begin{equation}
  \label{eqn:p_fail}
  p_{\text{fail}} = \sum_{\mcal{T}} {n \choose w_x, w_y, w_z} p_x^{w_x}p_y^{w_y}p_z^{w_z} (1-p)^{n-w}
\end{equation}
where $\mcal{T} = \lb w_x, w_y, w_z \; : \; w_x+w_y \leq t_x, w_z+w_y \leq t_z \rb$,  $t_i = \lfloor(d_i - 1)/2\rfloor$ and $p=p_x+p_y+p_z$.

\section{Single-parameter estimation}
\label{sec:singleparameterestimation}

We begin with an oriented Pauli channel $\Lambda_{\theta}$ which is a single-parameter generalization of dephasing noise. The action of the channel on one qubit can be described by a single Kraus operator, $A \lp \theta \rp$
\begin{equation}
  \Lambda_{\theta} \lp \rho \rp
  = \lp 1 - p \rp \rho + p A \lp \theta \rp \rho A \lp \theta \rp \label{eqn:oneKraus}.
\end{equation}
where
\begin{equation}
  A(\theta) = e^{-i \theta Y} Z e^{i \theta Y} .
\end{equation}
We do not know a priori the parameter $\theta$ and thus the effective channel, as perceived by our measurements of the stabilizer generators, is equivalent to that of $\Lambda_{\theta}$ under the Pauli twirl approximation~\cite{pta1}:
\begin{equation}
  \mcal{P} \ls \Lambda_{\theta} \rs \lp \rho \rp
   = \lp 1 - p \rp \rho  + p \cos^2 \theta X \rho X + p \sin^2 \theta Z \rho Z.
\end{equation}
This same expression can easily be derived by noting that
\begin{equation}
	e^{-i\theta Y} Z e^{i\theta Y} = Z \sin(2\theta) + X \cos(2\theta)
\end{equation}
then expanding the formula for $\Lambda_{\theta}$ and measuring the $X$ and $Z$ operators.

The goal of the protocol we describe below is to recover the parameter $\theta$ and substitute for the stabilizer operators in the standard basis $g_s$ new stabilizers that are aligned to exploit the asymmetry of the noise channel and maximize the lifetime of the code. This means applying $\hat{U}_t = e^{-i \hat{\theta} Y}$ to the stabilizers and to the code space (by physically applying a rotation to each qubit of the codeword), where $\hat{\theta}$ is our estimate of $\theta$ recovered after each round of error correction.

\subsection{Fixed dephasing angle}
\label{sec:fixed_dephasing_angle}
Our first case, and the one where we will demonstrate the greatest gain in code lifetime, is when the angle $\theta$ is fixed in time. We use a Bayesian estimator~\cite{sivia} that begins with uniform belief about the orientation channel $P_0(\theta; \theta_0) = 1 / \pi$ for $\theta \in [0, \pi)$. Every round of stabilizer error correction will diagnose and correct a total of $w_x$ bit-flip errors and $w_z$ phase-flip errors. Our knowledge of the dephasing angle at time $t+1$ includes information learned from the error weights, the previous configuration of the estimator, and the true value of $\theta_0$. Thus, we can express our belief about $\theta$ at time $t+1$ using Bayes' rule:
\begin{align}
  & P_{t+1}(\theta | w_x, w_z; \hat{\theta}_t, \theta_0) \nonumber \\
  & = \frac{P_{t+1}(w_x, w_z | \theta; \hat{\theta}_t,  \theta_0) P_{t+1}(\theta; \hat{\theta}_t, \theta_0)}{P_{t+1}(w_x, w_z)} \\
  & = \frac{P_{t+1}(w_x, w_z | \theta; \hat{\theta}_t,  \theta_0) P_{t+1}(\theta; \hat{\theta}_t, \theta_0)}{\int P_{t+1}(w_x, w_z | \theta; \hat{\theta}_t, \theta_0) P_{t+1}(\theta; \hat{\theta}_t, \theta_0) d\theta} \label{eqn:bayesupdate}
\end{align}
where $\theta$ is the random variable resulting from the update, $\hat{\theta}_t$ is the angle to which the stabilizers were configured in the previous round of error correction, and
\begin{align}
  P_{t+1}(w_x, w_z | \theta; \hat{\theta}_t, \theta_0)
  & = {n \choose w_x, w_z} p^{w_x + w_z} (1-p)^{n - w_x - w_z} \nonumber \\
  & \times \cos^{2w_x} \lp \lp \theta - \hat{\theta}_t \rp - \theta_0 \rp \nonumber \\
  & \times \sin^{2w_z} \lp \lp \theta - \hat{\theta}_t \rp - \theta_0 \rp .
\end{align}
After each update, we choose the new alignment of the stabilizers to be the maximum likelihood estimator of $\theta$:
\begin{equation}
  \hat{\theta}_t = \underset{\theta}{\mathrm{argmax}} \;\; P_t(\theta| w_x, w_z; \hat{\theta}_{t-1}, \theta_0) .
\end{equation}

This update rule has a particular advantage we can exploit to reduce the complexity of our estimator going forward. Note that when $w_x \geq 0$ and $w_z = 0$, $\hat{\theta}_t = \hat{\theta}_{t-1}$, which follows from the fact that $\cos^2 \lp \lp \theta - \hat{\theta}_t \rp - \theta_0 \rp$ has a single maximum on the interval $\ls 0, \pi \rp$ and is symmetric around it. Thus, we can perform our updates to the distribution of $\theta$ in bulk whenever we encounter a correctable $Z$ error by multiplying the distribution by the function
\begin{align}
  f_{t+1}(\theta, n_x; \hat{\theta}_t, \theta_0)
  & = \cos^{2n_x} \lp \lp \theta - \hat{\theta}_t \rp - \theta_0 \rp \nonumber \\
  &   \times \sin^{2} \lp \lp \theta - \hat{\theta}_t \rp - \theta_0 \rp .
\end{align}
where $n_x$ is the number of $X$ errors observed since the previous update. Each such step is followed by a renormalization of the posterior probability distribution $P_{t+1}(\theta| w_x, w_z; \hat{\theta}_t, \theta_0)$.

This yields an advantage for our numerical simulations as well. Since we need only update $\hat{\theta}_t$ on single $Z$ errors (as the $[[15, 1, 7/3]]$ code fails for $w_z > 1$), we can first sample the time until the next $Z$ error, then retroactively sample the number of $X$ errors to have occurred in the intermediate times, all the while accounting for the possibility of uncorrectable $X$ errors. This proves especially useful when our estimator is close to the optimal value $\theta$. In this case, $Z$ errors happen very infrequently, and the error rate is dominated by $O(p^4)$ terms, leading to many idle cycles of the simulation.

An important complication is that the distribution $P_t(\theta)$ cannot be stored in memory exactly and must be discretized. Let $P_t(\theta)$ be defined for the midpoints of the cells
\begin{equation}
  \ls 0, \pi \rp = \bigcup_{j=1}^{N-1} \ls \frac{j \pi}{N}, \frac{(j+1) \pi}{N} \rp.
\end{equation}
Thus, the difference between the true value $\theta_0$ and the maximum likelihood estimate $\hat{\theta}_t$ cannot necessarily be reduced to zero even at very long times, but may be as large as $\pi/2N$. This means that the rate of $Z$ errors can be, at best, suppressed to $p\sin^2(\pi/2N)$. Recall, however, that we do not need fully suppress $Z$ errors in order to take advantage of an asymmetric code. For very low effective $Z$ error-rates, the uncorrectable error rate in Eq.~\eqref{eqn:p_fail} is dominated by uncorrectable $X$ errors. If $p_z = cp_x^{t_x/t_z}$ for some $c > 0$, then
\begin{equation}
  p_{\text{fail}} = \lp c{n \choose t_z+1} + {n \choose t_x+1} \rp p_x^{t_x} + O \lp p_x^{t_x+1} \rp ,
\end{equation}
where $t_x=3$ and $t_z=1$ for the $[[15, 1, 7/3]]$ code. We can expect this very low rate only when $p_z = c p_x^2$ which is possible when  $\sin^2(\pi/2N) = O(p)$ or $N = O(1/p)$. Thus, we can only make optimal use of our estimate of $\theta$ when we take at least $N=O(1/p)$ partitions of $\ls 0, \pi \rp$.

The results of our simulations in Figure~\ref{fig:singleparamgraphs} show that for a fixed dephasing angle, our technique not only improves the code lifetime by a constant factor, but effectively increases the code distance. The mean lifetime of our adaptive $[[15, 1, 7/3]]$ code agrees with a power law of $p_{\text{fail}} = O\lp p^{3.99} \rp$ yielding an ``effective'' code distance of $6.98$. The next smallest CSS code that could yield a similar scaling is the $[[23, 1, 7]]$ code, but our adaptive technique outperforms even this code by a constant factor, thanks to the smaller number of physical qubits used.

\subsection{Drifting dephasing angle}

We now address the version of the previous problem with a parameter $\theta$ that is drifting in time. Consider a dephasing angle that at each time-step evolves via random Brownian motion:
\begin{align}
  \theta_{t+1}
  & = \theta_t + u \nonumber \\
  & \hspace{0.15in} \text{with } u \sim \mcal{N}(0, \kappa^2).
\end{align}
We can incorporate our knowledge of this drift into the estimator from the fixed angle case by convolving, at each time-step, our belief about $P_t(\theta_t)$ with the distribution of one step of the Brownian motion. That is, following every update in Eq.~\eqref{eqn:bayesupdate} we also apply
\begin{equation}
  P_{t+1}'(\theta_{t+1}) = P_{t+1}(\theta_{t+1}) * \lp \frac{1}{\kappa} \exp^{-\theta_{t+1}^2 / 2 \kappa^2} \rp ,
  \label{eqn:convolve}
\end{equation}
where $*$ denotes a convolution of the two functions defined as
\begin{equation}
  f(t) * g(t) = \int f(s)g(t-s)ds.
\end{equation}
We then take $\hat{\theta}_{t+1}$ to be the maximum likelihood estimate with respect to $P_{t+1}'(\theta_{t+1})$ instead.

Unlike in our simulations of the fixed angle case, we cannot make use of the retroactive $X$ error sampling that allowed us to simulate code lifetimes up to $10^{17}$ cycles.  We must therefore simulate every update straightforwardly. However, as Figure~\ref{fig:drift_trajectory} demonstrates, our estimator retains the property that the maximum likelihood estimate $\hat{\theta}_t$ only moves following $Z$ errors.

\begin{figure}
    \includegraphics[width=1.0\linewidth]{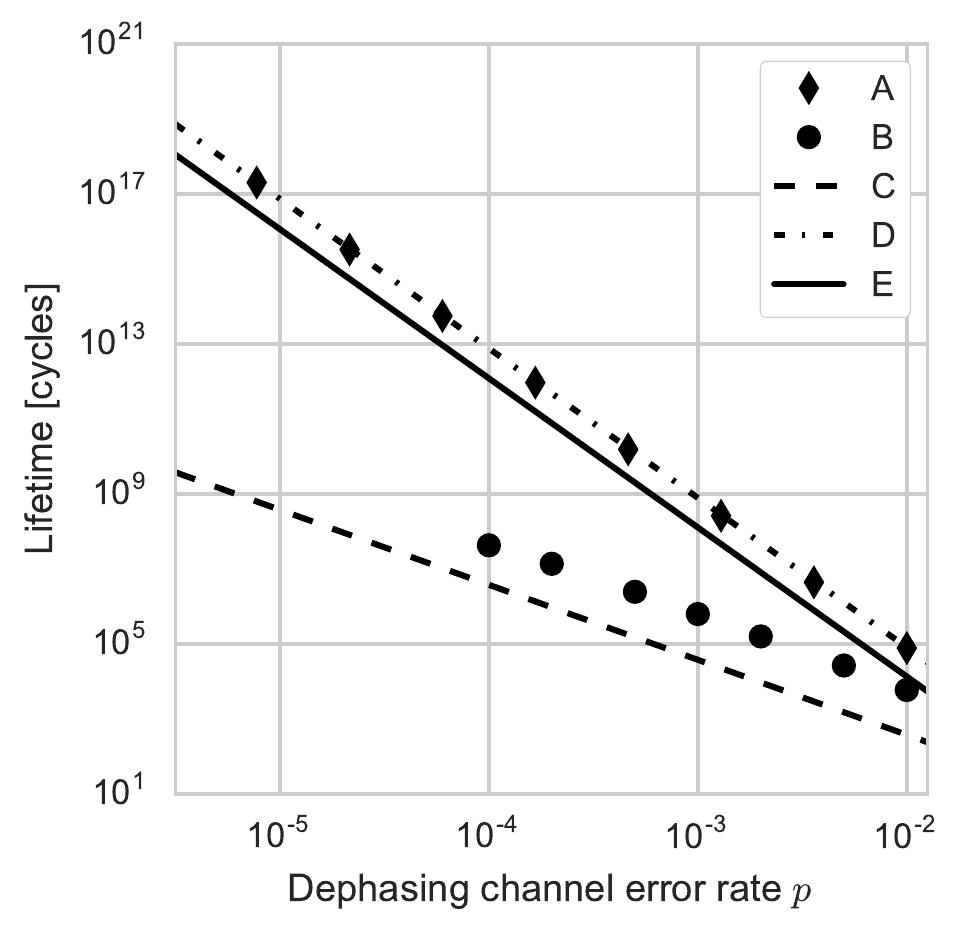}
    \caption{Performance of the adaptive stabilizers for the $[[15, 1, 7/3]]$ shortened Reed-Muller code in the presence of one-parameter generalized dephasing noise. \textbf{(A)} Mean lifetime of $5000$ samples ($100$ samples for $p=10^{-5}$) measured in terms of error correction cycles.  \textbf{(B)} Mean lifetime of $200$ samples in the presence of drifting dephasing noise ($\kappa^2=0.01$) measured in terms of error correction cycles. \textbf{(C)} Mode of the lifetime of the $[[15, 1, 7/3]]$ code without adaptive stabilizers. \textbf{(D)} Optimal lifetime of the $[[15, 1, 7/3]]$ code given perfect a priori knowledge of $\theta$. \textbf{(E)} Expected lifetime of the $[[23, 1, 7]]$ code. Error bars in this plot are sufficiently small to be disregarded.}
  \label{fig:singleparamgraphs}
\end{figure}

\begin{figure}
    \includegraphics[width=1.0\linewidth]{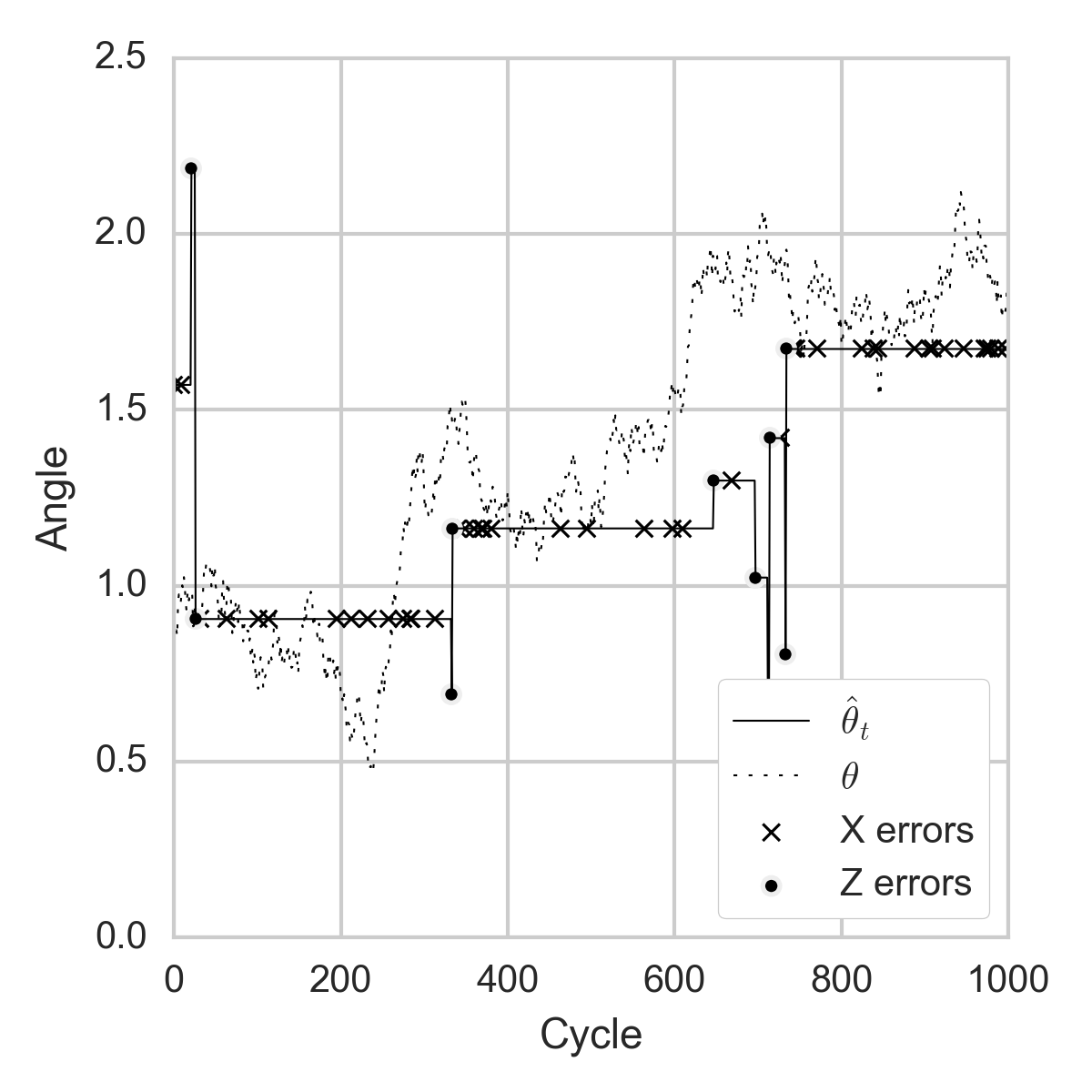}
  \caption{One run of our Bayesian estimator for $\theta$ drifting according to Brownian motion $\kappa^2 = 0.03$ and an error rate $p=0.003$.}
\label{fig:drift_trajectory}
\end{figure}

Figure~\ref{fig:singleparamgraphs} shows that using this modified estimator still yields a constant factor improvement to the lifetime of the $[[15, 1, 7/3]]$ code, but the lifetimes no longer scale with $O(1/p^4)$, the rate of uncorrectable $X$ errors. Instead, the drift in the channel effectively causes a constant but suppressed $Z$ error rate. This is because the the drift in $\theta_t$ implies that we cannot stay in the optimal configuration for very long.  This fact is reflected in our estimator by the convolution step, which widens any otherwise ``sharp" (i.e.: certain) belief about $\hat{\theta}_t$. From our simulations, the constant factor gain to the code lifetime for $\kappa^2=0.01$ is $6.33$.

\section{Multi-parameter estimation}
\label{sec:multiparameterestimation}

Although the single-angle dephasing model from the previous section is useful for studying the limits to which we can exploit a code's asymmetry, it is not very general, and we now turn our attention to the case of oriented Pauli channels. Recall that the Bloch matrix for such a channel can be written as a contraction in some non-standard basis:
\begin{equation}
  M = (1-2p) I + 2p Q^T \ls \begin{array}{ccc}
      k_1 & 0 & 0 \\
      0 & k_2 & 0 \\
      0 & 0 & k_3
    \end{array} \rs Q ,
\end{equation}
for some $Q \in SO(3)$ and its transpose $Q^T$. We choose this form for the Bloch matrix in order to highlight that our true objective in estimating $M$ is to find a rotation of the codespace and stabilizers such that the channel appears Pauli and allows for optimal use of the asymmetric code. Let $k_i$ be the \emph{eccentricities} of the oriented Pauli channel such that $k_1+k_2+k_3=1$, $p$ be the \emph{total error rate}, and $Q$ be the \emph{orientation} of the channel. Note that this channel, unlike the one in the previous section, can give rise to $Y$ errors. These can be corrected by any CSS code as simultaneous $X$ and $Z$ errors on the same qubit. Finally, we let $A$ be the matrix such that $M = (1-2p)I + 2pA$.

Note that we do not need to know $M$ to design the control unitary $\hat{U}_t$ for this scheme; knowledge of $A$ is sufficient to align the largest value of $k_i$ with the $X$ axis. In other words, the optimal choice of $\hat{U}_t$ does not depend on the total error rate $p$ (although estimating $p$ can also be done efficiently~\cite{fujiwara2014instantaneous}). If we let $\hat{Q}_t$ be the real orthonormal matrix associated with the unitary $\hat{U}_t$, then the effect of this counter-rotation on the Bloch matrix is $M = (1-2p) I + 2p \hat{Q}_t A \hat{Q}_t^T$, and it is clear that we should choose $\hat{Q}_t$ to diagonalize $A$ and order the eigenvalues such that $k_1 \geq k_3 \geq k_2$. This last details ensures that $Z$ errors are favored over $Y$ errors and $X$ errors are favored over both.

The performance of asymmetric codes against oriented Pauli channels with all $k_i > 0$ is limited, since even with optimal orientation, they exhibit non-zero $Z$ and $Y$ error rates. These channels cannot yield improvements to the effective code distance like those in the dephasing case, but they can yield improvements to the coefficient of the uncorrectable error rate, and thereby to the code lifetime.

Let $C_{\text{opt}}$ be the constant factor improvement to the lifetime yielded by the $[[15, 1, 7/3]]$ code for an oriented Pauli channel, given perfect a priori knowledge of the optimal orientation $Q_{\text{opt}}$. $C_{\text{opt}}$ then is given by the ratio of the expected lifetimes
\begin{equation}
  C_{\text{opt}} \simeq \frac{\mbE \ls t(Q_{\text{opt}}) \rs}{\int \mbE \ls t(Q) \rs dQ}
  \label{eqn:Copt}
\end{equation}
where $\mbE \ls t(Q) \rs$ is the lifetime at orientation $Q$ given by $1/p_{\text{fail}}(Q)$, the uncorrectable error rate. We take $\int dQ$ to be an integral over the Haar measure for $SO(3)$~\cite{simon1996representations}. For $k_1, k_2, k_3 \gg p$, the uncorrectable error is dominated by $ZZ$, $YY$, and $ZY$ errors. The approximation becomes
\begin{align}
  p_{\text{fail}}(Q)
  & = {15 \choose 0, 2, 0} p_y^2 + {15 \choose 0, 0, 2} p_z^2 + {15 \choose 0, 1, 1} p_yp_z , \nonumber \\
  & = 105(p_y + p_z)^2 , \nonumber \\
  & = 105(p-p_x)^2 .
\label{eqn:Creasoning}
\end{align}
The $X$ error rate is $p_x = p \vec{e}_1^T QDQ^T \vec{e}_1$ where $\vec{e}_i$ are the standard basis vectors and $D$ is the diagonal matrix with diagonal elements $k_i$. Note that for $Q$ with columns $\vec{q}_i$, the $X$ error rate reduces to $p \vec{q}_1^TD\vec{q}_1$, where $\vec{q}_1$ follows a uniform distribution on $SO(3)$ given by
\begin{equation}
  \vec{q} = \ls \begin{array}{c}
    u \\
    \sqrt{1-u^2} \sin \lp 2 \pi v \rp \\
    \sqrt{1-u^2} \cos \lp 2 \pi v \rp
  \end{array} \rs
  \hspace{0.25in}
  \lb \begin{array}{l}
    u \sim \text{unif} \ls -1, 1 \rs , \\
    v \sim \text{unif} \ls 0, 1 \rs .
  \end{array} \right.
\end{equation}
We can thus derive the distribution on $k_x$:
\begin{equation}
	k_x = k_1u^2 + k_2(1-u^2)\sin^2 \lp 2 \pi v \rp + k_3(1-u^2)\cos^2 \lp 2 \pi v \rp.
\end{equation}
To take the Haar integral needed for the average lifetime, we must calculate
\begin{align}
	\lefteqn{\int \mbE \ls t(Q) \rs dQ} \nonumber\\
	& = \int_{-1}^1 \int_0^1 \frac{1}{105p^2(1-k_x)^2} dv du , \nonumber\\
	& \leq \int_{-1}^1 \frac{1}{105p^2 \lp 1-k_1u^2 - k_3(1-u^2) \rp^2} du , \nonumber\\
	& = \int_{-1}^1 \frac{1}{105p^2 \lp (1-k_3)-(k_1-k_3)u^2 \rp^2} du , \nonumber\\
	& = \frac{\tanh^{-1} \lp \sqrt{(k_1-k_3)/(1-k_3)} \rp}{210p^2(1-k_3)\sqrt{(1-k_3)(k_1-k_3)}} \nonumber\\
	& \hspace{0.15in} + \frac{1}{210p^2(1-k_3)(1-k_1)} . 
\end{align}
where in the second line we've used that $k_3 \leq k_2$. For the channel of eccentricities $(0.7, 0.2, 0.1)$, the optimal and average lifetimes become
\begin{align}
  \mbE \ls t(Q_{\text{opt}}) \rs & \simeq \frac{11.11}{105p^2} , \\
  \int \mbE \ls t(Q) \rs dQ & \lesssim \frac{2.71}{105p^2} .
\end{align}
Thus, the optimal improvement to the code lifetime that we can expect is $C_{\text{opt}} \gtrsim 4.01$. A similar calculation for the $[[31, 6, 7/5]]$ code yields $C_{\text{opt}} \gtrsim 9.34$.

\subsection{Non-degenerate grids for Bayesian inference}

We again will use Bayesian inference to estimate the relevant parameters of the matrix $A$. Note, however, that were we to learn the rates of $X$, $Y$ and $Z$ errors, we could only determine the diagonal elements of matrix $A$. This is a problem acknowledged by~\cite{omkar} and~\cite{ferrie}. In~\cite{omkar} the authors propose ``toggling" the codespace to estimate off-diagonal elements of the Bloch matrix by rotating the codespace and pre-processing the encoded state. Although their method is applicable to much more general multi-qubit process matrices, it cannot be performed in-situ without incurring significant costs by moving the codespace away from the optimum.

The solution we propose is to sample the space of Bloch matrices randomly, and track the posterior probabilities for a field of $N$ sampling points in the parameter space of $A$. This method of approximating the likelihood function is similar to our point-wise estimate in section~\ref{sec:fixed_dephasing_angle} which was supported on regularly-spaced sampling points in $\mbR$. We sample each point $X$ in our randomized grid as follows. First, we will sample eccentricities $x_1$, $x_2$, $x_3$ according to the distribution
\begin{align}
          \Pr \lb x_1 \rb = 1 , & \hspace{0.25in} \text{for } x_1 \in \ls 0, 1 \rs , \nonumber\\
 	  \Pr \lb x_2 |x_1 \rb = 1/(1-x_1) , & \hspace{0.25in} \text{for } x_2 \in \ls 0, 1-x_1 \rs ,
\end{align}
and $x_3 = 1 - x_1 - x_2$. We take $D_X$ to be the diagonal matrix with diagonal elements $x_1$, $x_2$, and $x_3$. Note that the average channel in our grid will have eccentricities $(1/2, 1/4, 1/4)$, which means that our prior distribution assumes some asymmetry in the channel. Next, we sample an orthonormal basis $Q_X$ from the Haar measure over $SO(3)$ (for which efficient methods exist~\cite{graphicsgems}). Finally, we let our sampling point be
\begin{equation}
	X  = Q_X^T D_X Q_X .
	\label{eqn:randomorientedpauli}
\end{equation}

We now relate the number of sampling points $N$ to the optimal performance of our asymmetric codes. We do this in two steps:  first, by bounding the minimum distance of a fixed channel $A$ to the closest element in a randomized grid, and second, by bounding the rate of uncorrectable errors using the minimum distance of the grid. The first step is summarized in the following Lemma, the proof of which can be found in the Appendix.

\begin{lem}[Number of sampling points versus minimum distance]
\label{lem:gridsize}
Let $X_i$ be $N$ i.i.d. copies of randomly oriented Pauli channels according to~\eqref{eqn:randomorientedpauli}, and let $A$ be a fixed channel with eccentricities $(a_1, a_2, a_3)$. Then we have that
\begin{equation}
  \Pr \lb Z < \e \rb \geq 1 - \lp 1 - \frac{\e^5}{2^{12}\sqrt{2}3^3 a_1^2a_2a_3} \rp^N, 
  \label{eqn:gridbound}
\end{equation}
where $Z=\min_i \left\| X_i - A \right\|_2$ and $\left\| Y \right\|_2 = \sqrt{\Tr \ls Y^T Y \rs}$ is the Frobenius norm.
\end{lem}

For $\e \ll 1$ the bound in Lemma~\ref{lem:gridsize} scales as
\begin{equation}
  \Pr \lb \min_i \left\| X_i - A \right\|_2 < \e \rb \geq N \frac{\e^5}{2^{12}\sqrt{2}3^3 a_1^2a_2a_3}.
\end{equation}
Thus, to ensure that a channel in our random grid approximates $A$ to precision $\e$ in the Frobenius norm with probability $p_{A}$, we need to populate our random grid with a very large number of sampling points:
\begin{equation}
  N = \frac{2^{12}\sqrt{2}3^3 a_1^2a_2a_3}{\e^5p_{A}}.
\end{equation}
We provide this estimate for the explicit purpose of characterizing the computational resources required to implement our scheme. If the adaptive encoding procedure outlined in this text is to be implemented on classical computing architecture operating in the very low-latency setting of stabilizer quantum error correction, then it must be possible to achieve the gains we describe without unreasonable overhead.

We can also use this grid size spacing to estimate code performance. For the oriented Pauli channel with eccentricities $k_1 = 0.7$, $k_2 = 0.2$, $k_3 = 0.1$, and with $30,000$ sampling points, we can expect
\begin{equation}
  \Pr \lb Z < \e \rb \geq 6.92 \e^5.
\end{equation}
The estimate we found for the code performance in Eq.~\eqref{eqn:Copt} assumes complete and exact knowledge of the channel matrix $A$. What we've found in the derivation above, however, is a bound on our ability to learn an approximation $X$ of the matrix $A$.  As such, the lifetime of our code in practice will not be bounded by the optimal performance of Eq.~\eqref{eqn:Copt}, but rather by the quality of this approximation. We now derive the relationship between the Frobenius norm distance between $X$ and $A$ and the code performance.

Let $x_{11}$ be the first diagonal element of $X$.  Then we know that
\begin{align}
	\lefteqn{\left| x_{11} - k_1 \right| \leq \left\| X - A \right\|_2} \\
	& \hspace{0.1in} \Longrightarrow \hspace{0.1in} \left| 1- x_{11} \right| \leq \left| 1 - k_1 \right| + \left\| X - A \right\|_2 .
\end{align}
We identify the lifetime coefficient $C$ as the ratio of $p_{\text{fail}}$ for random and optimal orientations, which can be bounded from below by
\begin{equation}
	C \geq \frac{(2/3)^{t_z}}{(1-k_1 + \left\| X - A \right\|_2)^{t_z}}.
	\label{eqn:code_perf}
\end{equation}
We call this lower bound the \emph{expected code performance}, and will make use of it as a baseline against which to compare the numerical results of the following section.

\section{Numerical results}
\label{sec:numericalresults}

In Figures~\ref{fig:codeperfs1} and~\ref{fig:codeperfs2} we graph the results of $1000$ trials of our adaptive estimation and error correction method for the $[[15, 1, 7/3]]$ and the $[[ 31, 6, 7/5 ]]$ codes, at $N=30,000$ sampling points. We also plot the expected code performance as a function of the Frobenius norm using our bound derived above. Recall that the expected code performance serves as a lower bound, and our results indicate that it is reasonably tight.

When our adaptive technique is applied, the distribution of lifetimes for the $[[31, 6, 7/5]]$ code changes from the usual geometric distribution to a more uniform distribution with a high proportion of trials lasting from $2$ to $10$ times the expected unadaptive lifetime. We explain this behavior in the following way: the $31$-qubit code fails at a rate of $p^3$ for weight $3$ $Z$ errors, and thus the average number of updates to the posterior at the expected time of failure scales with $1/p^3$. The $15$-qubit code, on the other hand, fails at a rate of $p^2$ and so has, on average, learned a worse estimate of $A$ than the $31$-qubit code at the expected time of failure. This suggests that the code distance has an additional effect beyond simply decreasing the baseline rate of uncorrectable errors. The data in Figures~\ref{fig:codeperfs1} and~\ref{fig:codeperfs2} suggest that larger codes benefit proportionally more from our adaptive technique, as they are able to learn the channel parameters more accurately before the probability of an uncorrectable error becomes appreciably high.

\begin{figure}
   \includegraphics[width=1.0\linewidth]{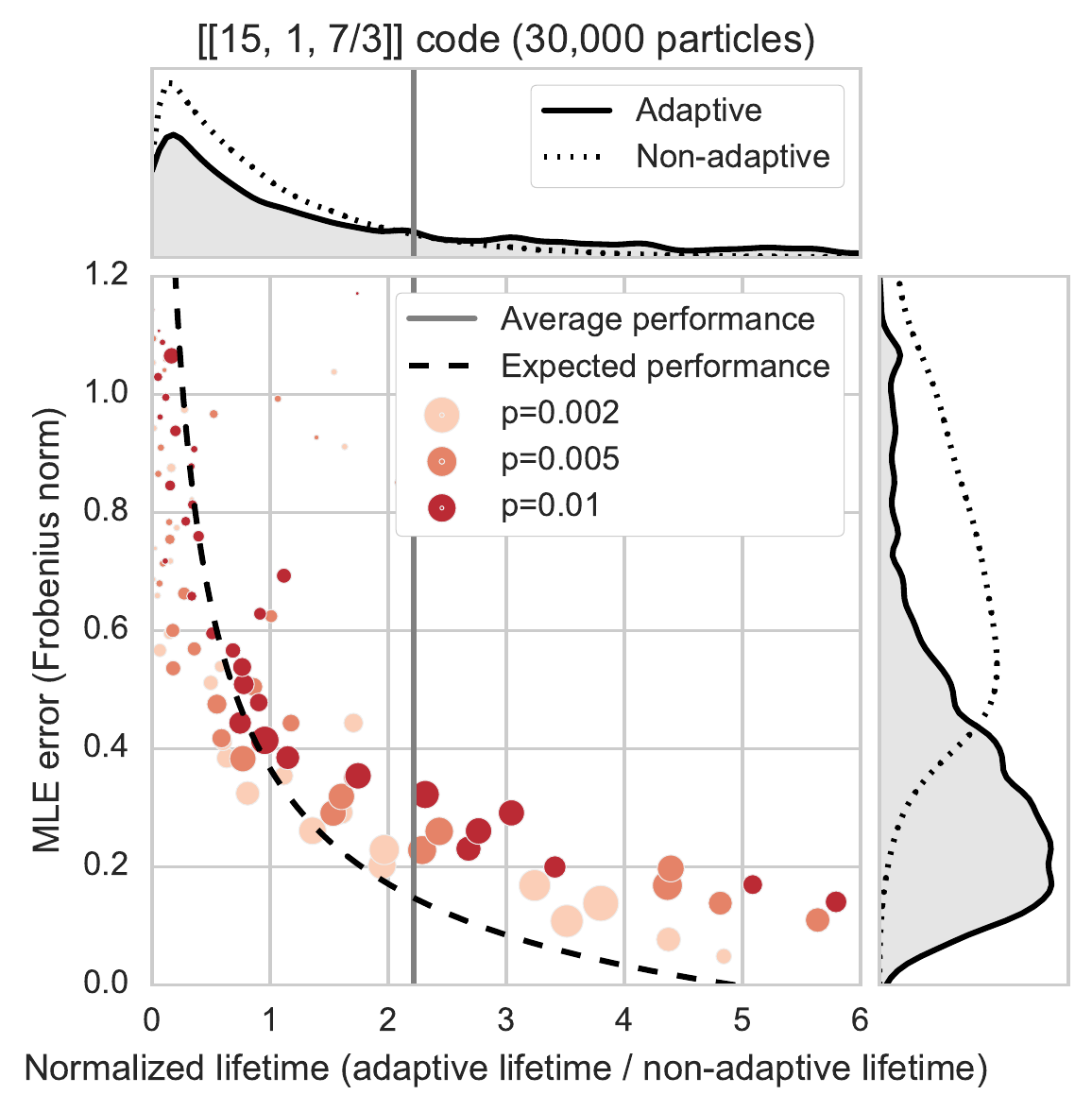}
  \caption{Performance (as defined in Eq.~\eqref{eqn:Copt}) of the adaptive stabilizers for the $[[15, 1, 7/3]]$ shortened Reed-Muller code with $30,000$ sampling points in the presence of a fixed unital channel of eccentricities $(0.7, 0.2, 0.1)$ and an unknown orientation. $300,000$ error correction simulations were performed to obtain this graph. In each trial, the adaptive correction procedure is performed until an uncorrectable error is encountered. The MLE (maximum likelihood estimator) error is taken as the Frobenius distance between the Bloch matrix of the channel given by the MLE and the target matrix at the time of the uncorrectable error.  For each error rate, the size of the marker represents the number of trials that yielded the given normalized lifetime. The distribution of both of these trial properties is also summarized in the top and left histograms.}
  \label{fig:codeperfs1}
\end{figure}

\begin{figure}
   \includegraphics[width=1.0\linewidth]{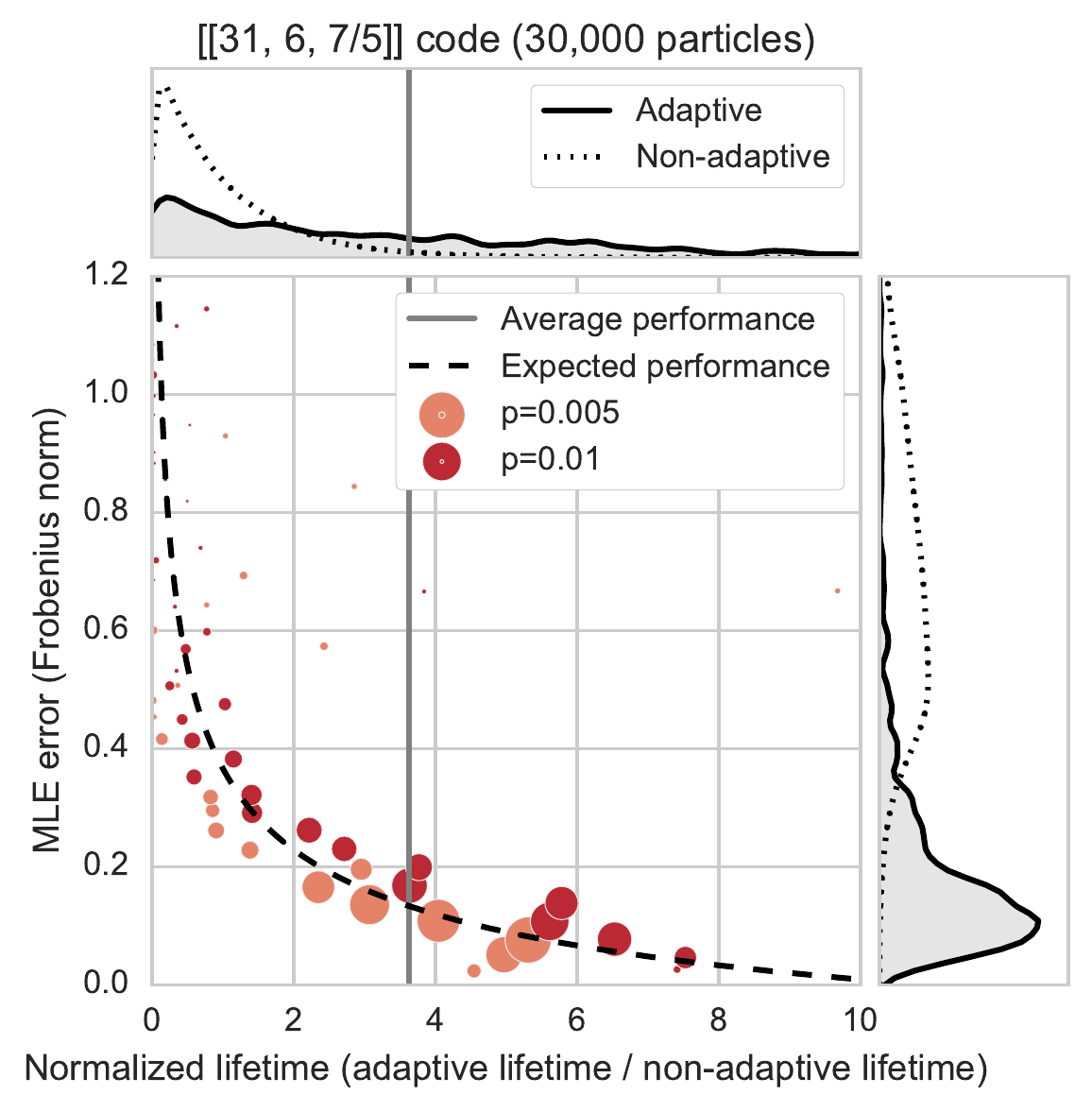}
  \caption{Performance (as defined in Eq.~\eqref{eqn:Copt}) of the adaptive stabilizers for the $[[ 31, 6, 7/5 ]]$ code with $30,000$ sampling points in the presence of a fixed unital channel of eccentricities $(0.7, 0.2, 0.1)$ and an unknown orientation. The $x$ and $y$ axes, as well as the size of the plotted points, are defined as in Figure~\ref{fig:codeperfs1}.}
  \label{fig:codeperfs2}
\end{figure}

In Table~\ref{tab:codeperfs} we summarize the improvement to the lifetime and the Frobenius error at the time of failure of each code. In Figure~\ref{fig:particle_effects} we illustrate that the minimum Frobenius distance of a random channel does indeed follow a power law of $N^{1/5}$, but our bound from Eq.~\eqref{eqn:gridbound} overestimates the true grid spacing by a significant factor. In the same figure we demonstrate the decay of the improvement to the code lifetime with decreasing number of sampling points $N$. Surprisingly, our method still yields some benefit even with only $10$ sampling points. This is likely due to the fact that when sampling random channels and random grid points, the probability of generating a grid point close to $A$ is still high enough that some trials have this property and inherit disproportionately long lifetimes, thereby raising the average.

  \begin{table}
    \renewcommand\tabcolsep{6pt}
    \newcolumntype{C}{>{\centering\arraybackslash}m{1.65cm}}
    \begin{center} \begin{tabular}{|C|C|C|C|C|}
        \hline Code & Number of sampling points & Normalized lifetime $C$ & Frobenius error \\
        \hline
        \hline $[[15, 1, 7/3]]$ & $2500$ & 2.04 & 0.37 \\
        \hline $[[15, 1, 7/3]]$ & $30,000$ & 2.17 & 0.41 \\
        \hline $[[31, 6, 7/5]]$ & $2500$ & 3.17 & 0.20 \\
        \hline $[[31, 6, 7/5]]$ & $30,000$ & 3.68 & 0.24 \\
        \hline
    \end{tabular}\end{center}
    \caption{Average gains made to the code lifetime from numerical simulations of $1000$ samples at each error rate. The Frobenius error is calculated at the time of code failure.}
    \label{tab:codeperfs}
  \end{table}

\begin{figure}[h]
  \centering
    \includegraphics[width=1.0\linewidth]{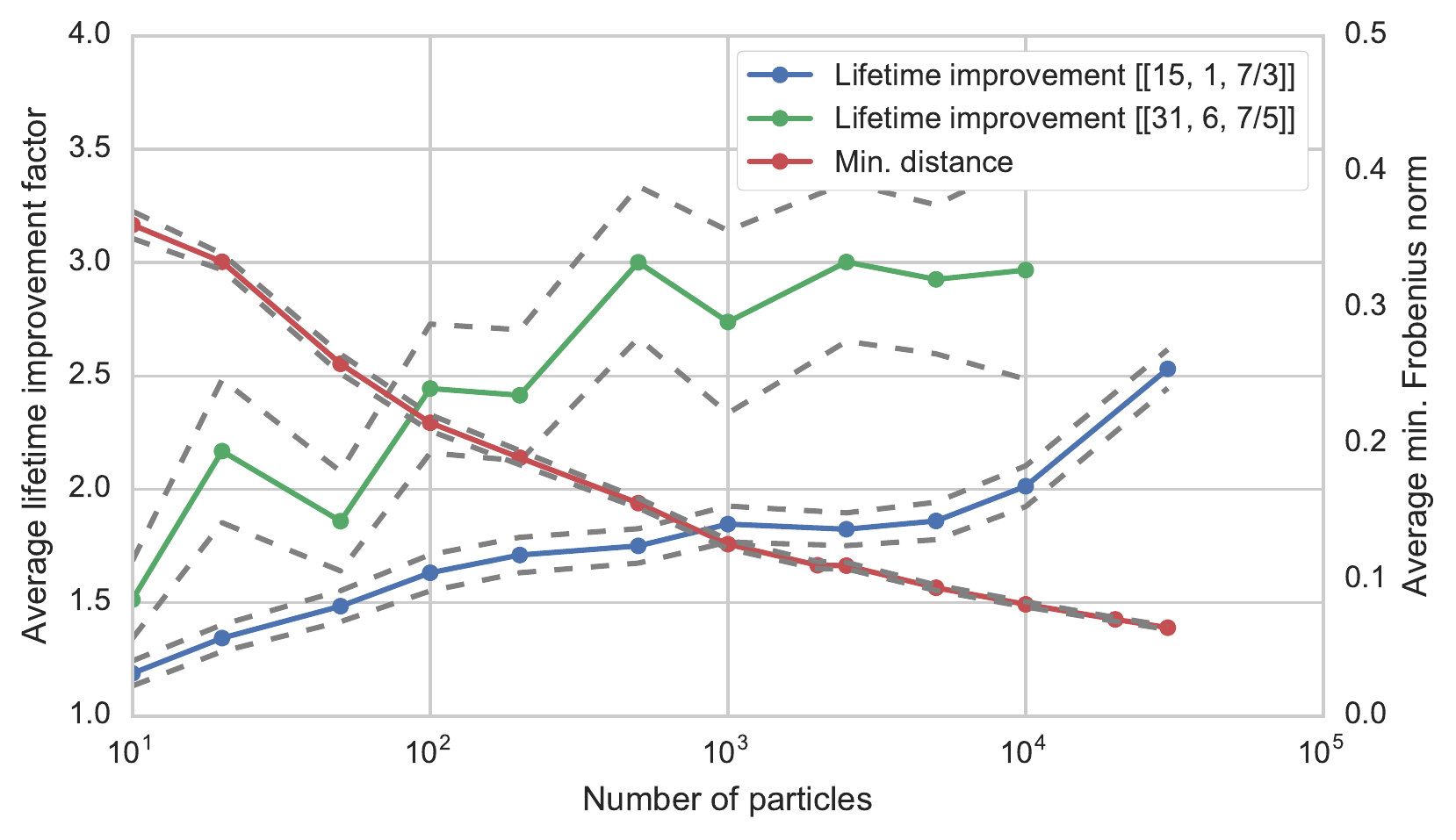}
  \caption{Effect of number of sampling points on lifetime factor for the $[[15, 1, 7/3]]$  and $[[31, 6, 7/5]]$ codes, and on the average Frobenius norm distance to the nearest grid element, for a random channel of eccentricity $(0.7, 0.2, 0.1)$. The dotted lines represent the standard error around the mean.}
  \label{fig:particle_effects}
\end{figure}

\section{Discussion}
\label{sec:discussion}

By adaptively changing the codespace and stabilizers of asymmetric error correcting codes, we've demonstrated an in-situ method for boosting their performance in the presence of asymmetric noise channels. Our method provides a constant factor improvement to the code lifetime in all cases, and in the case of fixed dephasing noise along an unknown axis even results in a higher ``effective" code distance. Since our control unitaries on the physical qubits are functions only of the syndrome statistics, our method is also able to track noise parameters that drift in time.

It is interesting to note that in the case of oriented Pauli channels, the performance of the code is dominated by our ability to learn the direction of the largest channel eccentricity. This allows us to align our code in a way that matches this eccentricity to $X$ errors, which our code is best equipped to correct. The gains we make by learning the orientation of the second largest eccentricity are much less significant, since the code's ability to correct $Y$ errors is only slightly worse than for $Z$ errors.

In addition to our method of adaptively rotating the codespace and stabilizers, one might also consider re-encoding the logical quantum state as information about the noise channel becomes more refined. For example, if the noise channel is discovered to be highly asymmetric for a sustained period of time, then one could convert the stabilizer code into one that tolerates fewer $Z$ and $Y$ errors and thereby save on the total number of physical qubits and stabilizer measurements. Several works~\cite{conversion1, guillaume, conversion2} already demonstrate how to convert between these codes fault-tolerantly. Furthermore, one might consider correcting for some types of errors less often, as is done in~\cite{asymmetric1} in the presence of a strongly biased channel.

It would seem that our method could also extend to unital channels in general. In this case, we would model the process matrix $M$ not with an eigenvalue decomposition but with a singular value decomposition and our task would be to estimate both the left and right bases of $M$. Equivalently, we could retain our oriented Pauli channel estimator and concatenate it with another unitary channel. This is similar to~\cite{ferrie} where the authors consider a Pauli channel in the standard basis concatenated with a unitary rotation. In both cases, a straightforward application of our randomized grid would incur a further $\e^3$ scaling to the bound on the number of sampling points in Eq.~\eqref{eqn:gridbound}, significantly driving up the minimum number of sampling points needed to yield lifetime improvements.  It might also be possible to go beyond unital channels, but this might require new methods to estimate the ``displacement'' of the noise channel.

Finally, a very natural extension of our method would be to consider separate noise channels and separate control unitaries for each qubit in the code. This case could be addressed by simply running parallel channel estimators for each qubit, with updates applied only for errors that occur on that qubit. Note that in this case the learning rate drops by $1/n$ on average, and each estimator converges much more slowly. However, it is likely that the channels would be correlated based on the qubit topology, an assumption which could be used to decrease the number of free parameters needed.

\begin{acknowledgements}
The authors acknowledge useful conversations with Kung-Chuan Hsu, Jos\'e Raul Gonzalez Alonso, Shengshi Pang and Christopher Cantwell.  This research was supported in part by the ARO MURI under Grant No. W911NF-11-1-0268; NSF Grant No. CCF-1421078; and an IBM Einstein Fellowship at the Institute for Advanced Study.
\end{acknowledgements}

\bibliographystyle{abbrv}
\bibliography{AllRefs}

\begin{appendices}
\section{Appendix}
\label{sec:lemmas}

\begin{lem}[Oriented Pauli channels as contractions]
\label{lem:Mrewrite}
Given the oriented Pauli channel $\Lambda_M$ (Def.~\ref{def:orientedpauli}), we can always write the process matrix $M$ as
\begin{equation}
	M = (1-2p) I + 2p Q^T_U \ls \begin{array}{ccc}
      k_1 & 0 & 0 \\
      0 & k_2 & 0 \\
      0 & 0 & k_3
    \end{array} \rs Q_U.
\end{equation}
\end{lem}
\begin{proof}
Recall that $\Lambda_M = \Lambda_{U^{\dag}} \circ \Lambda_D \circ \Lambda_{U}$. Consider the action of $\Lambda_D$ on a qubit state:
\begin{align}
	\Lambda_D(\rho)
	& = (1-p) \rho + p_x X \rho X + p_y Y \rho Y + p_z Z \rho Z \nonumber\\
	& = (1-p) \frac{I + \vec{r} \cdot \vec{\s}}{2} + p_x \frac{I + X\vec{r} \cdot \vec{\s}X}{2} \nonumber\\
	& \hspace{0.15in} + p_y \frac{I + Y\vec{r} \cdot \vec{\s}Y}{2} + p_z \frac{I + Z\vec{r} \cdot \vec{\s}Z}{2} \nonumber\\
	& = \frac{I}{2} + (1-p) \frac{\vec{r} \cdot \vec{\s}}{2} \nonumber\\
	& \hspace{0.25in} + (p_x - p_y - p_z) \frac{r_1 \s_x}{2} \nonumber\\
	& \hspace{0.25in} + (-p_x + p_y - p_z) \frac{r_2 \s_y}{2} \nonumber\\
	& \hspace{0.25in} + (-p_x - p_y + p_z) \frac{r_3 \s_z}{2} \nonumber\\
	& = \frac{I}{2} + (1-2p) \frac{\vec{r} \cdot \vec{\s}}{2} \nonumber\\
	& \hspace{0.25in} + \frac{2p_xr_1 \s_x + 2p_yr_2 \s_y + 2p_zr_3 \s_z}{2} \nonumber\\
	& = \frac{I + \lp \ls A \rs \vec{r} \rp \cdot \vec{\s}}{2},
\end{align}
where
\begin{equation}
A = (1-2p) I + 2p \ls \begin{array}{ccc}
      k_1 & 0 & 0 \\
      0 & k_2 & 0 \\
      0 & 0 & k_3
    \end{array} \rs
\end{equation}
where $p_x = pk_1$, $p_y=pk_2$ and $p_z=pk_3$. It remains only to identify that $\Lambda_U$ and $\Lambda_{U^{\dag}}$ act as transformations $Q_U$ and $Q_{U^{\dag}}$ on $\vec{r}$ before and after the channel $\Lambda_D$ to complete the proof.
\end{proof}

\begin{lemmaproof}[1][Number of sampling points versus minimum distance]
Let $X_i$ be $N$ i.i.d. copies of randomly oriented Pauli channels according to~\eqref{eqn:randomorientedpauli}, and let $A$ be a fixed channel with eccentricities $(a_1, a_2, a_3)$. Then we have that
\begin{equation}
  \Pr \lb Z < \e \rb \geq 1 - \lp 1 - \frac{\e^5}{2^{12}\sqrt{2}3^3 a_1^2a_2a_3} \rp^N, 
  \label{eqn:gridbound}
\end{equation}
where $Z=\min_i \left\| X_i - A \right\|_2$ and $\left\| Y \right\|_2 = \sqrt{\Tr \ls Y^T Y \rs}$ is the Frobenius norm.
\end{lemmaproof}
\begin{proof}
Consider the variable $Z \sim \min_i \left\| X_i - A \right\|_2$. The cumulative distribution function of $Z$ is
\begin{align}
  \Pr \lb Z < \e \rb
  & = 1 - \prod_i  \lp 1 - \Pr \lb \left\| X_i - A \right\|_2 < \e \rb \rp \nonumber\\
  & = 1 - \lp 1 - \Pr \lb \left\| X - A \right\|_2 < \e \rb \rp^N \label{eqn:mindist}, 
\end{align}
and let
\begin{align}
  \mu(X)
  & = \left\| X - A \right\|_2 \nonumber\\
  & = \left\| Q_A^TXQ_A - D_A \right\|_2 \nonumber\\
  & = \left\| Q_{XA}^TD_XQ_{XA} - D_A \right\|_2, 
\end{align}
where in the second line we've conjugated both matrices by $Q_A$ and $Q_{XA} = Q_XQ_A$ is also Haar random~\cite{simon1996representations}, since the Haar measure is invariant under left or right conjugation. Let $E_D = \lb \left| D_X - D_A \right| \leq \e_D \rb$ where $\left| \cdot \right|$ is the absolute value, i.e., the event in which the eigenvalues of $X$ and $A$ are similar. Conditioning on this event we see that,
\begin{align}
  \lefteqn{\Pr \lb \mu(X) \leq \e \rb} \nonumber\\
  & = \Pr \lb \mu(X) \leq \e | E_D \rb \Pr \lb E_D \rb \nonumber\\
  & \hspace{0.15in} + \Pr \lb \mu(X) \leq \e | \neg E_D \rb \Pr \lb \neg E_D \rb \nonumber\\
  & \geq \Pr \lb \mu(X) \leq \e | E_D \rb \Pr \lb E_D \rb .
\end{align}
Thus, we can consider only those channels $X$ for which $E_D$ is true. Let $\{\vec{q}_i\}$ be the column vectors of $Q_{XA}$, and let $x_i$ and $a_i$ be the diagonals of $D_X$ and $D_A$, respectively. From $E_D$ we know that $|x_i - a_i| \leq \e_D$. The Frobenius norm can then be simplified to be
\begin{align}
  \lefteqn{\left\| Q_{XA}^TD_XQ_{XA} - D_A \right\|_2} \nonumber\\
  & = \left\| \sum_{i=1}^3 \lp x_i \vec{q}_i \vec{q}_i^T  - a_i \vec{e}_i \vec{e}_i^T \rp \right\|_2 \nonumber\\
  & \leq \sum_{i=1}^3 \left\| x_i \vec{q}_i \vec{q}_i^T - a_i \vec{e}_i \vec{e}_i^T \right\|_2 \nonumber\\
  & \leq \sum_{i=1}^3 \lp a_i \left\|\vec{q}_i \vec{q}_i^T - \vec{e}_i \vec{e}_i^T \right\|_2  + \e_D \rp \nonumber\\
  & \leq \sum_{i=1}^3 \lp \sqrt{2} a_i \left\| \vec{q}_i - \vec{e}_i \right\|_2 + \e_D \rp .
\end{align}
where in the second line we've used the convexity of the Frobenius norm, in the third line we've used previous our bounds on the distance between $x_i$ and $a_i$, and in the fourth line we've used the result from Lemma~\ref{lem:frobnorm}. Let $A_i = \lb \sqrt{2} a_i \left\| \vec{q}_i - \vec{e}_i \right\|_2 + \e_D < \e/3 \rb$.  Then we must have
\begin{align}
  \lefteqn{\Pr \lb \mu(X) < \e | E_D \rb} \nonumber\\
  & \geq \Pr \lb A_1 \cap A_2 \cap A_3 \rb \nonumber\\
  & = \Pr \lb A_1 \rb \Pr \lb A_2 | A_1 \rb \Pr \lb A_3 | A_1 \cap A_2 \rb .
\end{align}
Note that the marginal distribution of each $\vec{q}_i$ is uniform over the unit sphere and can be constructed using~\cite{graphicsgems},
\begin{equation}
  \vec{q} = \ls \begin{array}{c}
    u \\
    \sqrt{1-u^2} \sin \lp 2 \pi v \rp \\
    \sqrt{1-u^2} \cos \lp 2 \pi v \rp
  \end{array} \rs
  \hspace{0.15in}
  \lb \begin{array}{l}
    u \sim \text{unif} \ls -1, 1 \rs , \\
    v \sim \text{unif} \ls 0, 1 \rs .
  \end{array} \right.
\end{equation}
Thus, for $\vec{q}_1$ to lie within $\e$ of $\vec{e}_1$ in the Frobenius norm, we need only that $u \in \ls 1 - \e^2/2, 1 \rs$, and for $u$ uniformly distributed over $\ls -1, 1 \rs$ this means that
\begin{equation}
  \Pr \lb \left\| \vec{q}_1 - \vec{e}_1 \right\|_2 \leq \e \rb = \e^2/2 .
\end{equation}
Next, the vector $\vec{q}_2$ is sampled uniformly from the equator of vectors orthogonal to $\vec{q}_1$. Given that $\left\| \vec{q}_1 - \vec{e}_1 \right\|_2 \leq \e$, we are guaranteed that $\vec{e}_2$ lies in the band of latitude $\e$ around this equator. Thus, the probability of $\vec{q}_2$ lying within $\e$ of $\vec{e}_2$ is given by the ratio of the area of the circle of radius $\e$ on the unit sphere to that of the band. For $\e \ll 1$ this is,
\begin{equation}
  \Pr \lb \left\| \vec{q}_2 - \vec{e}_2 \right\|_2 \leq \e \rb = \frac{\int_0^{\e} 2\pi r dr}{\int_{1 - \e}^{1 + \e} 2\pi r dr} = \frac{\e}{2} . \\
\end{equation}
Finally, given that $\vec{q}_1$ and $\vec{q}_2$ each lie within $\e$ of $\vec{e}_1$ and $\vec{e}_2$ respectively, the vector $\vec{q}_3$ lies within $\e$ of $\vec{e}_3$ with probability $1/4$. Altogether,
\begin{align*}
  \lefteqn{\Pr \lb \left\| \vec{q}_1 - \vec{e}_1 \right\|_2 \leq \e \rb = \e^2/2} \\
  & \Longrightarrow & 
  \Pr \lb A_1 \rb = \frac{1}{8a_1^2}\lp \e/3 - \e_D \rp^2,  \\
  \lefteqn{\Pr \lb \left\| \vec{q}_2 - \vec{e}_2 \right\|_2 \leq \e \rb = \e/2} \\
  & \Longrightarrow &
  \Pr \lb A_2 | A_1 \rb = \frac{1}{4a_2} \lp \e/3 - \e_D \rp,  \\
  \lefteqn{\Pr \lb \left\| \vec{q}_3 - \vec{e}_3 \right\|_2 \leq \e \rb = 1/4} \\
  & \Longrightarrow &
  \Pr \lb A_3 | A_1 \cap A_2 \rb = \frac{1}{2a_3} \lp 1/4 - \e_D \rp .
\end{align*}
If we choose $\e_D = \e/6$, then we have altogether that
\begin{equation}
  \Pr \lb \mu(X) < \e | E_D \rb \geq \frac{\e^3}{2^{10}\sqrt{2}3^3 a_1^2a_2a_3}. 
\end{equation}

It now remains to bound $\Pr \lb E_D \rb$ to finish the proof. The event $E_D$ is characterized by the probabilities,
\begin{align}
  \Pr \lb \left| X_1 - a_1 \right| \leq \e \rb
  & \geq \e,  \\
  \Pr \lb \left| X_2 - a_2 \right| \leq \e | \left| X_1 - a_1 \right| \leq \e \rb
  & \geq \e .
\end{align}
Finally, if $|X_1 - a_1| \leq \e$ and $|X_2 - a_2 | \leq \e$, then we are guaranteed that $|X_3 - a_3| \leq 2\e$. Putting this all together, we can use the chain rule to see that
\begin{equation}
  \Pr \lb \left| D_X - D_A \right| \leq \e \rb \geq (\e/2)^2 .
\end{equation}
Thus, we have our bound on the spacing of a random channel $X$:
\begin{equation}
  \Pr \lb \mu(X) < \e \rb \geq \frac{\e^5}{2^{10}\sqrt{2}3^3 a_1^2a_2a_3}.
\end{equation}

In our bound on the minimum distance Eq.~\eqref{eqn:mindist} this becomes
\begin{equation}
  \Pr \lb Z < \e \rb \geq 1 - \lp 1 - \frac{\e^5}{2^{10}\sqrt{2}3^3 a_1^2a_2a_3} \rp^N.
\end{equation}
\end{proof}

\begin{lem}[Frobenius distance of outer products]
\label{lem:frobnorm}
Given $\vec{x}, \vec{y} \in \mbR^n$, where $\vec{x}$ and $\vec{y}$ are unit vectors, we can bound the Frobenius distance of their outer product by the Frobenius distance of the vectors themselves,
\begin{equation}
	\left\| \vec{x} \vec{x}^T - \vec{y} \vec{y}^T \right\|_2 \leq 2 \left\| \vec{x} - \vec{y} \right\|_2 .
\end{equation}
\end{lem}
\begin{proof}
Recall that for any real matrix $A$, $\left\| A \right\|_2 = \sqrt{\Tr \ls A A^T \rs}$, thus
\begin{align}
	\left\| \vec{x} \vec{x}^T - \vec{y} \vec{y}^T \right\|^2_2
	& = \Tr \ls \vec{x} \vec{x}^T \vec{x} \vec{x}^T - \vec{x} \vec{x}^T \vec{y} \vec{y}^T \right. \nonumber\\
	& \hspace{0.35in} \left. - \vec{y} \vec{y}^T \vec{x} \vec{x}^T + \vec{y} \vec{y}^T \vec{y} \vec{y}^T \rs \nonumber\\
	& = \langle \vec{x}, \vec{x} \rangle^2 + \langle \vec{y}, \vec{y} \rangle^2 - 2 \langle \vec{x}, \vec{y} \rangle^2 \nonumber\\
	& = \sqrt{2} \lp 1 - \langle \vec{x}, \vec{y} \rangle^2 \rp \nonumber\\
	& \leq 2\sqrt{2} \lp 1 - \langle \vec{x}, \vec{y} \rangle \rp \nonumber\\
	& = \sqrt{2} \left\| \vec{x} - \vec{y} \right\|_2^2 .
\end{align}
\end{proof}

\begin{lem}[Bound on diagonal matrix elements]
\label{lem:diagelts}
Let $D_A$ and $D_B$ be diagonal matrices of entries $a_i$ and $b_i$ respectively. Let $X = Q^TD_AQ$ for some orthogonal matrix $Q$ with columns $\vec{q}_i$. We then have that
\begin{equation}
	\left| \ls X \rs_{ii} - b_i \right| \leq \left\| X - D_B \right\|_2 .
\end{equation}
\end{lem}
\begin{proof}
First, recall the form of the Frobenius norm for any matrix $Y$:
\begin{equation}
	\left\| Y \right\|_2 = \sqrt{\displaystyle\sum_{i=1}^n \displaystyle\sum_{j=1}^n \left| y_{ij} \right|^2} .
\end{equation}
In our case, this becomes
\begin{align}
	\left\| X - D_B \right\|_2
	& = \sqrt{\displaystyle\sum_{i=1}^n \displaystyle\sum_{j=1}^n \left| \ls X \rs_{ij}-b_i\d_{ij} \right|^2} \nonumber\\
	& \geq \sqrt{\left| \ls X \rs_{ii}-b_i \right|^2} \nonumber\\
	& = \left| \ls X \rs_{ii} - b_i \right| .
\end{align}

\end{proof}

\end{appendices}

\end{document}